\documentclass[final,12pt,3p]{elsarticle}

\usepackage{graphicx}
\usepackage{amssymb}
\usepackage{amsthm}
\usepackage{amsmath}
\usepackage{algorithm}
\usepackage{algorithmic}
\usepackage{tikz}
\usepackage{lineno}
\usepackage{hyperref}
\usepackage{verbatim}

\usepackage{caption}
\usetikzlibrary{shapes}
\usepackage{subcaption}
\usepackage{geometry}
\usepackage{mathrsfs} 
\usepackage{array}
\usepackage{multirow}
\usepackage{booktabs} 
\usepackage{adjustbox}
\usepackage[title]{appendix}

\biboptions{sort&compress}
\hypersetup{
colorlinks=true,
linkcolor=red,
filecolor=green,
urlcolor=green,
anchorcolor=green,
citecolor=cyan,
pdftitle={Overleaf Example},
pdfpagemode=FullScreen,
}
\newtheorem{corollary}{Corollary}
\newtheorem{proposition}{Proposition}

\newcommand{\ot}[1]{\frac{\mathrm{d} {#1}}{\mathrm{d}t}}

\newcommand{\annuity}[2]{\bar{a}^{#1}_{\overline{#2}|}}
\newcommand{\annuityinf}[1]{\bar{a}^{#1}_{\overline{\infty}|}}

 \usepackage{xcolor}

 \usepackage[figcolor=white]{todonotes}

\newcommand{\Tmatthias}[1]{\todo[inline, color=blue!40]{Matthias: #1}}

\begin{document}
	
	\begin{tikzpicture}[remember picture,overlay]
		\node[anchor=north east,inner sep=20pt] at (current page.north east)
		{\includegraphics[scale=0.2]{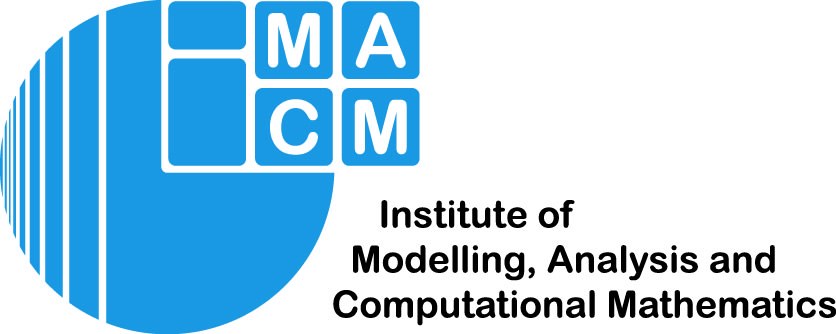}};
	\end{tikzpicture}

\begin{frontmatter}

\title{Actuarial Analysis of an Infectious Disease Insurance\\ based on an SEIARD Epidemiological Model}

\author[BUW,MAIS]{Achraf Zinihi}
\ead{a.zinihi@edu.umi.ac.ma} 

\author[BUW]{Matthias Ehrhardt\corref{Corr}}
\cortext[Corr]{Corresponding author}
\ead{ehrhardt@uni-wuppertal.de}

\author[MAIS]{Moulay Rchid Sidi Ammi}
\ead{rachidsidiammi@yahoo.fr}

\address[BUW]{University of Wuppertal, Applied and Computational Mathematics,\\
Gaußstrasse 20, 42119 Wuppertal, Germany}

\address[MAIS]{Department of Mathematics, AMNEA Group, Faculty of Sciences and Techniques,\\
Moulay Ismail University of Meknes, Errachidia 52000, Morocco}



\begin{abstract}
The growing number of infectious disease outbreaks, like the one caused by the SARS-CoV-2 virus, underscores the necessity of actuarial models that can adapt to epidemic-driven risks. Traditional life insurance frameworks often rely on static mortality assumptions that fail to capture the temporal and behavioral complexity of disease transmission. In this paper, we propose an integrated actuarial framework based on the SEIARD epidemiological model. This framework enables the explicit modeling of incubation periods and disease-induced mortality.\\
We derive key actuarial quantities, including the present value of annuity benefits, payment streams, and net premiums, based on SEIARD dynamics. We formulate a prospective reserve function and analyze its evolution throughout the course of an epidemic. Additionally, we examine the forces of infection, mortality, and removal to assess their impact on epidemic-adjusted survival probabilities. Numerical simulations implemented via a nonstandard finite difference (NSFD) scheme illustrate the model’s applicability under various parameter settings and insurance policy assumptions.
\end{abstract}


\journal{} 



\end{frontmatter}


\section{Introduction}\label{S1}
Mathematical modeling has become an essential tool for studying infectious disease dynamics. It offers a quantitative framework for understanding how epidemics spread, how they can be controlled, and how they behave in the long term.
One of the most widely used approaches is the compartmental model, which divides the population into epidemiological states, such as susceptible, exposed, infected, and recovered \cite{Yang2024, Simeonov2023, Zinihi2025S, Chang2022, Zinihi2025FDE, LuciaSanz2023}. 
First formalized in 1927 by Kermack and McKendrick \cite{kermack1927}, these models have since become the foundation of modern epidemiological theory.
These models are typically expressed as systems of ordinary differential equations (ODEs) and are used to estimate key epidemiological quantities such as the basic reproduction number $\mathcal{R}_0$ and the herd immunity threshold \cite{Azimaqin2025, Hethcote2000, Brauer2019}. 
The SEIRD model extends the classical SIR framework by incorporating a latency period through the exposed compartment and distinguishing between recovery and disease-induced mortality.
This makes it especially suitable for modeling diseases with delayed symptom onset and measurable fatality, such as SARS, Ebola, and, most recently, SARS-CoV-2 \cite{Giordano2020, Zhao2020}.

Meanwhile, actuarial science offers a rigorous mathematical basis for assessing and managing uncertain future events related to life insurance, pensions, and financial risk. Classical actuarial models rely heavily on life tables and survival models, in which constant mortality rates and independence between lives are standard assumptions \cite{Dickson2019}. 
These frameworks enable the pricing of life insurance products, the estimation of reserves, and the assessment of solvency using probability theory and financial mathematics techniques \cite{Kaas2008}.
However, large-scale health crises such as the pandemic caused by the SARS-CoV-2 virus have exposed significant limitations in these traditional models, especially when mortality rates fluctuate rapidly and systemic risk affects large populations simultaneously.
In such contexts, researchers have proposed integrating epidemic-aware compartmental models into actuarial analysis to enable more responsive premium structures, dynamic reserve adjustments, and real-time solvency management \cite{Feng2011, Hainaut2020}.

In a related study, Nkeki and Iroh \cite{Nkeki2024} developed an epidemiological model for designing communicable disease insurance policies. 
The model divided the population into five groups: susceptible individuals, exposed individuals, individuals who were infective and under treatment, deceased individuals, and recovered individuals who could rejoin the susceptible group and continue paying premiums. 
The infective group included active cases and deceased policyholders. 
The study focused on deriving the probability and cumulative distribution functions for policyholders' risk statuses. 
Using actuarial techniques, the authors computed various insurance-related quantities, including the financial obligations of the insurer and the insured.
In another study, \cite{Hainaut2020} proposed a pandemic risk model designed for the actuarial valuation of insurance products offering healthcare and death benefits. 
This model relied on a deterministic framework, which was presented as an efficient alternative to the classical SIR model. 
It was used to describe the early dynamics of the SARS-CoV-2 virus in several European countries, including Belgium, Germany, Italy, and Spain. 
The model permitted analytical tractability in calculating fair pure premiums. 
Two stochastic extensions were introduced: the first replaced calendar time with a gamma stochastic clock to capture uncertainty in the epidemic peak; 
the second incorporated Brownian motion and jump processes to account for erratic fluctuations and local resurgences in infection counts.
Similarly, \cite{Francis2023}  extended traditional multi-state models in life insurance to incorporate the effects of epidemic contagion. 
By modifying transition intensities within a Markov framework, the study captured infectious disease transmission at the individual level in a manner analogous to compartmental epidemiological models. 
This enabled a consistent actuarial evaluation of personal risk and insurance reserves under epidemic conditions.

\cite{Chatterjee2008} constructed a stochastic Markov model to simulate an individual's progression through health risk factors, such as obesity, diabetes, and hypertension, which can lead to ischemic heart disease or stroke.
\cite{Macdonald2005} extended this modeling approach to incorporate multiple critical illnesses and examined its implications for underwriting and premium rating in critical illness insurance. This included the effects of genetic predispositions and behavioral factors.
\cite{Nkeki2024Epi} proposed a compartmental model that explicitly captures infection, recovery, and reinfection. They applied actuarial techniques to evaluate insurance premiums, reserves, and financial obligations.
\cite{Feng2021}  provided a broader review of classical epidemiological models and their actuarial applications, including product design, epidemic reserving, and healthcare resource planning.
\cite{Chernov2021} studied premium pricing in SEIR models with migration and vaccination, highlighting how disease dynamics and public health interventions affect insurance costs.
\cite{Francis2023} extended multi-state Markov models to incorporate contagion effects at the individual level, enabling a more accurate valuation of reserves under epidemic conditions.
These studies represent only a subset of the growing body of literature at the intersection of epidemiological modeling and actuarial science. Interested readers may refer to the references therein and broader surveys such as \cite{Dickson2019, Feng2011, Hainaut2020, Nkeki2024, Francis2023, Macdonald2005, Nkeki2024Epi, Feng2021, Chernov2021}.

In this study, we examine a SEIARD compartmental model \cite{de2020seiard}, 
which is an extension of the classical SIR framework that incorporates asymptomatic infections and disease-induced mortality. 
The population is divided into six compartments: 
susceptible ($S$), exposed ($E$), symptomatic infected ($I$), asymptomatic infected ($A$), recovered ($R$), and deceased ($D$). 
This structure provides a more realistic representation of infectious disease dynamics, particularly in insurance modeling contexts where mortality and varying infection outcomes directly impact policyholder status and claims.
Including asymptomatic carriers and a separate mortality class enhances the model’s relevance for actuarial applications.
This enables the study of how epidemic progression impacts life insurance liabilities, premium calculations, and reserve dynamics.

By modeling the transitions that policyholders make between health states over time, we can quantify the likelihood that they will become infected, recover, or die due to the disease. We can also link these transitions to financial outcomes, such as claim payments, reserve accumulation, and premium adjustments. This modeling approach is particularly relevant for health and life insurance products because the occurrence of an epidemic introduces significant uncertainty in mortality rates and policyholder behavior. In this paper, we use the SEIARD model to derive net premium expressions, compute prospective reserves, and evaluate the sensitivity of actuarial quantities to key epidemiological parameters. Through this analysis, we aim to build on our previous work.

The remainder of this paper is organized as follows. 
Section~\ref{S2} introduces the SEIARD compartmental model, which is used to describe epidemic dynamics. This section also includes a well-posedness analysis to ensure the mathematical validity of the system. 
Section~\ref{S3} focuses on the model's actuarial applications, including deriving formulas for premiums and benefits, computing insurance reserves, and examining the effects of epidemic parameters, such as the force of infection, mortality and removal, on financial outcomes. 
A nonstandard finite difference (NSFD) scheme is introduced to discretize the proposed SEIARD model in Section~\ref{S4}.
Section~\ref{S5} presents numerical simulations that illustrate the actuarial implications of the model under various epidemic scenarios.
Finally, Section~\ref{S6} concludes the paper with a summary of the findings and potential directions for future research.

\section{SEIARD Model Description}\label{S2}
In this section, we present the mathematical framework used to model the transmission dynamics of an infectious disease, as well as its actuarial implications. We use a SEIARD compartmental model, in which the population is divided into six health states.

\subsection{Mathematical Model with Vital Dynamics}\label{S2.1}
We consider an SEIARD epidemiological model with vital dynamics to describe the evolution of an infectious disease in a structured population.
The model incorporates constant birth (or recruitment) at rate $\Lambda$, and natural death at rate $\mu$ across all living compartments.
The force of infection is driven by both symptomatic and asymptomatic individuals. 
The latter contribute at a reduced transmission rate controlled by the parameter $\kappa \in [0,1]$.
Upon exposure, individuals progress to the infectious phase at rate $\alpha$, with a proportion $p$ becoming symptomatic, while the remaining $1-p$ remain asymptomatic.
Both symptomatic and asymptomatic individuals may recover or die from the disease at respective rates.
The deceased class, D(t), accumulates disease-related deaths and is not included in the \textit{total living population} $N_L(t)$.

The transmission parameters employed in the SEIARD model are summarized in Table~\ref{Tab1}.

\begin{table}[H]
\centering
\setlength{\tabcolsep}{0.8cm}
\caption{Transmission parameters for the proposed SEIARD model.}\label{Tab1}
\adjustbox{max width=\textwidth}{
\begin{tabular}{cc}
\hline 
\textbf{Symbol} & \textbf{Description} \\
\hline \hline 
$\Lambda$ & Recruitment rate (e.g.\ birth or immigration) \\
\hline
$\beta$ & Transmission rate\\
\hline
$\mu$ & Natural death rate \\
\hline
$\kappa$ & Relative infectivity of asymptomatic individuals \\
\hline
$\alpha$ & Rate of progression from exposed to infected\\
\hline
$p$ & Proportion developing symptomatic infection\\
\hline
$\gamma_I$, $\gamma_A$ & Recovery rates for $I$ and $A$\\
\hline
$\delta_I$, $\delta_A$ & Disease-induced death rates from $I$ and $A$\\
\hline
\end{tabular}
}
\end{table}

Figure~\ref{F1} illustrates the compartmental structure of the SEIARD model, highlighting key transitions between health states, such as infection, recovery, and disease-induced death.

\begin{figure}[H]
\centering
\begin{tikzpicture}[node distance=3.5cm]
\node (S) [rectangle, draw, minimum size=1cm, fill=cyan!30] {S};
\node (E) [rectangle, draw, minimum size=1cm, fill=orange!40, right of=S] {E};
\node (I) [rectangle, draw, minimum size=1cm, fill=red!40, right of=E] {I};
\node (A) [rectangle, draw, minimum size=1cm, fill=red!30, above of=I, yshift=-1.5cm] {A};
\node (R) [rectangle, draw, minimum size=1cm, fill=green!30, right of=I] {R};
\node (D) [rectangle, draw, minimum size=1cm, fill=blue!30, right of=R] {D};
\draw [->, thick, >=latex, line width=1pt] (-1.5,0) -- ++(S) node[midway,above]{$\Lambda$};
\draw [->, thick, >=latex, line width=1pt] (S) -- (E) node[midway,above]{$\beta S \Bigl( \frac{I + \kappa A}{N_L} \Bigr)$};
\draw [->, thick, >=latex, line width=1pt] (E) -- (I) node[midway,above]{$\alpha p E$};
\draw [->, thick, >=latex, line width=1pt] (I) -- (R) node[midway,above]{$\gamma_I I$};
\draw [->, thick, >=latex, line width=1pt] (E.north) -- (3.5,2) -- (6.5,2) node[midway,above]{$\alpha (1-p) E$};
\draw [->, thick, >=latex, line width=1pt] (7.5,1.9) -- (10.5,1.9) node[midway,below]{$\gamma_A A$} -- (R.north);
\draw [->, thick, >=latex, line width=1pt] (7.5,2.1) -- (14,2.1) node[midway,above]{$\delta_A A$} -- (D.north);
\draw [->, thick, >=latex, line width=1pt] (7.1,-0.5) -- (7.1,-1.5) -- (14,-1.5) node[midway,above]{$\delta_I I$} -- (D.south);
\draw [->, thick, >=latex, line width=1pt] (S.south) -| (0,-1.3) node[near end,left]{$\mu S$};
\draw [->, thick, >=latex, line width=1pt] (E.south) -| (3.5,-1.3) node[near end,left]{$\mu E$};
\draw [->, thick, >=latex, line width=1pt] (A.south) -| (7,0.7) node[near end,left]{$\mu A$};
\draw [->, thick, >=latex, line width=1pt] (R.east) -- (12,0) node[midway,above]{$\mu A$};
\draw [->, thick, >=latex, line width=1pt] (6.9,-0.5) -- (6.9,-1.3) node[near end,left]{$\mu I$};
\end{tikzpicture}
\captionof{figure}{Transmission pathways in the proposed SEIARD model with vital dynamics.}\label{F1}
\end{figure}
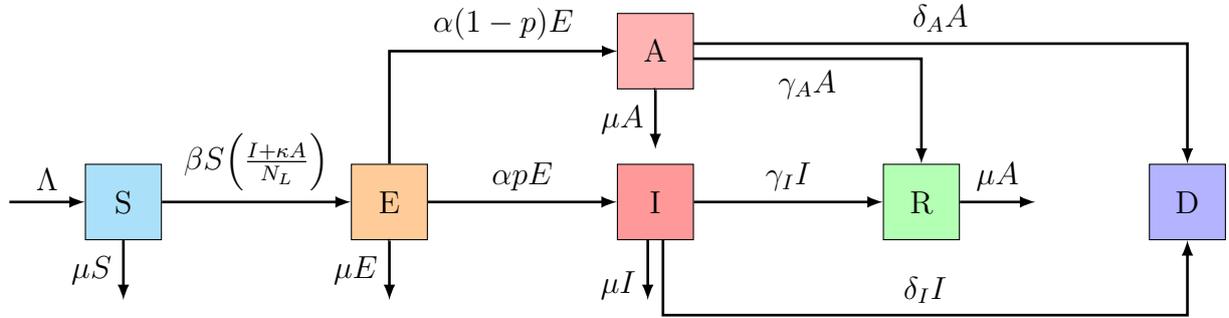

Let $N_L(t) = S(t) + E(t) + I(t) + A(t) + R(t)$ denote the \textit{total living population}. 
Then, the SEIARD model is given by the following system:

\begin{equation}\label{E2.1}
\left\{\begin{aligned}
\ot{S} &= \Lambda - \beta S \Bigl( \frac{I + \kappa A}{N_L} \Bigr) - \mu S, \\
\ot{E}&= \beta S \Bigl( \frac{I + \kappa A}{N_L} \Bigr) - (\alpha + \mu) E, \\
\ot{I} &= p\alpha E - (\gamma_I + \delta_I + \mu) I, \\
\ot{A}&= (1 - p)\alpha E - (\gamma_A + \delta_A + \mu) A, \\
\ot{R} &= \gamma_I I + \gamma_A A - \mu R, \\
\ot{D} &= \delta_I I + \delta_A A,
\end{aligned}\right.
\end{equation}
where
\begin{equation}\label{E2.2}
S(0) = S_0, \quad E(0) = E_0, \quad I(0) = I_0, \quad A(0) = A_0, \quad R(0) = R_0, \quad D(0) = D_0.
\end{equation}

From an actuarial perspective, these six compartments in \eqref {E2.1} play significantly different roles in an insurance model. 
First, people susceptible or exposed to infection during an epidemic ($S$, $E$) form a market share that could contribute premiums to an insurance fund in exchange for coverage of medical expenses if they become infected.
Next, during an epidemic outbreak, infected policyholders ($I$, $A$) benefit from claim payments to cover medical expenses provided by the insurance fund.
In our model, recovered persons ($R$) become immune and do not return to the susceptible class $S$. However, they continue to pay premiums to the insurance company.
After an insured person dies, i.e.\ transits to the compartment $D$, the beneficiaries they designated may receive a death benefit to cover funeral and burial expenses.
Note that once the insurance fund has been established, interest will accrue on the unpaid reserves at a fixed interest rate \cite{Feng2011}.
Figure~\ref{F2} summarizes the actuarial roles of the different compartments described above.

\begin{figure}[H]
\centering
\begin{tikzpicture}[node distance=3.5cm]
\node (I) [rectangle, draw, minimum size=1cm, fill=red!40] {I};
\node (A) [rectangle, draw, minimum size=1cm, fill=red!30, below of=I, yshift=1.5cm] {A};
\node (Ins) [rectangle, rounded corners, draw, minimum size=1cm, minimum width=1.5cm, fill=gray!20, right of=I, xshift=3cm, yshift=-1cm] {Insurance};
\node (E) [rectangle, draw, minimum size=1cm, fill=orange!40, right of=Ins, xshift=2cm] {E};
\node (S) [rectangle, draw, minimum size=1cm, fill=cyan!30, above of=E, yshift=-1.8cm] {S};
\node (R) [rectangle, draw, minimum size=1cm, fill=green!30, below of=E, yshift=1.8cm] {R};
\node (D) [rectangle, draw, minimum size=1cm, fill=blue!30, left of=Ins, xshift=0.4cm, yshift=2cm] {D};
\draw [-, thick, >=latex, line width=1pt] (E) -- (10.5,-1);
\draw [-, thick, >=latex, line width=1pt] (S) -| (11.2,0.7) -| (10.7,-1);
\draw [-, thick, >=latex, line width=1pt] (R) -| (11.2,-2.7) -| (10.7,-1);
\draw [->, thick, >=latex, line width=1pt] (10.5,-1) -- (Ins) node[midway,above]{\footnotesize Premium} node[midway,below]{\footnotesize Payment};
\draw [->, thick, >=latex, line width=1pt] (7.8,-1.2) -- (7.5,-1.2);
\draw [-, thick, >=latex, line width=1pt] (7.8,-1.2) -- (7.8,-1.9);
\draw [-, thick, >=latex, line width=1pt] (7.8,-1.9) -- (5.2,-1.9) node[midway,below]{\footnotesize Investment};
\draw [-, thick, >=latex, line width=1pt] (5.2,-1.9) -- (5.2,-1.2);
\draw [-, thick, >=latex, line width=1pt] (5.2,-1.2) -- (5.5,-1.2);
\draw [-, thick, >=latex, line width=1pt] (Ins) -- (4.5,-1);
\draw [-, thick, >=latex, line width=1pt] (4.5,-1) -- (4.5,1) node[yshift=-0.4cm,right]{\footnotesize Death} node[yshift=-0.8cm,right]{\footnotesize Benefit};
\draw [->,thick, >=latex, line width=1pt] (4.5,1) -- (D);
\draw [-, thick, >=latex, line width=1pt] (4.5,-1) -- (1.3,-1) node[midway,above]{\footnotesize Hospitalization} node[midway,below]{\footnotesize Benefit};
\draw [->, thick, >=latex, line width=1pt] (1.3,0) -- (I) ;
\draw [->, thick, >=latex, line width=1pt] (1.3,0) -| (1.3,-2) -- (A) ;
\end{tikzpicture}
\captionof{figure}{Insurance dynamics among compartments S, E, I, A, R, and D.}\label{F2}
\end{figure}
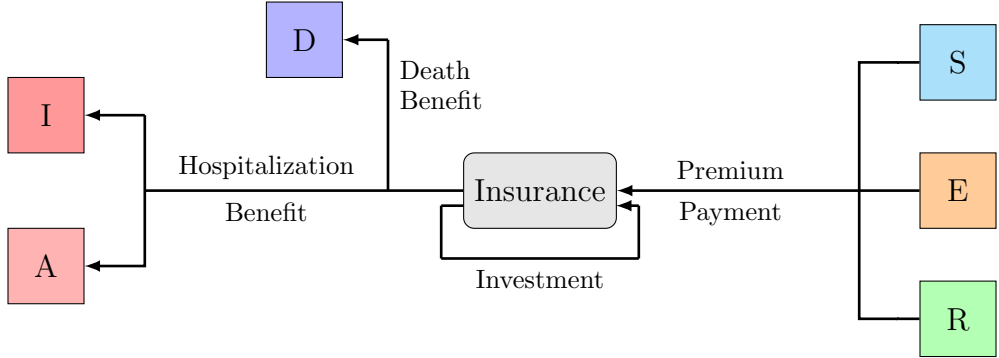

\subsection{Simplified Model without Vital Dynamics}\label{S2.3}

Although the model \eqref{E2.1} 
accounts for vital dynamics, many actuarial and epidemiological applications focus on short-term epidemic scenarios where demographic processes, such as birth and natural death, have a negligible influence compared to disease-induced effects. 
This assumption is particularly useful for analyzing insurance products affected by sudden epidemic shocks over short time periods.

In actuarial modeling, the main focus is often on quantifying the short-term impact of epidemic-driven mortality on life insurance reserves and capital requirements. During these short periods, background demographic changes occur slowly and have minimal influence on key financial indicators. Including vital dynamics requires additional parameters, such as birth and natural mortality rates, which are either unavailable or unnecessary for pricing and reserving during acute outbreaks.
For these reasons, the actuarial-epidemiological literature typically removes vital dynamics and focuses solely on the transmission and progression of the disease.
This simplification enhances mathematical tractability and allows for clearer interpretation of results relevant to insurance and risk management.

Therefore, we consider a 
simplified SEIARD model in which the recruitment ('birth') rate  $\Lambda$ and the natural death rate $\mu$ are both set to zero.
In this case, the \textit{total population}
\begin{equation*}
  N(t) = S(t) + E(t) + I(t) + A(t) + R(t) + D(t),    
\end{equation*}
remains constant because there are no births or non-disease-related deaths.
However, the \textit{total living population}
\begin{equation*}
   N_L(t) = S(t) + E(t) + I(t) + A(t) + R(t),
\end{equation*}
is not constant, since it decreases over time due to disease-induced mortality.
Therefore, under the assumptions described in Table~\ref{Tab1} and the compartmental structure illustrated in Figures~\ref{F1} and \ref{F2}, the 
\textit{simplified model} is given by
\begin{equation}\label{E2.3}
\left\{\begin{aligned}
\ot{S} &= - \beta S \Bigl( \frac{I + \kappa A}{N_L} \Bigr), \\
\ot{E}&= \beta S \Bigl( \frac{I + \kappa A}{N_L} \Bigr) - \alpha E, \\
\ot{I} &= p\alpha E - (\gamma_I + \delta_I) I, \\
\ot{A} &= (1 - p)\alpha E - (\gamma_A + \delta_A) A, \\
\ot{R} &= \gamma_I I + \gamma_A A, \\
\ot{D} &= \delta_I I + \delta_A A,
\end{aligned}\right.
\end{equation}
supplied with the initial data
\begin{equation}\label{E2.4}
S(0) = S_0, \quad E(0) = E_0, \quad I(0) = I_0, \quad A(0) = A_0, \quad R(0) = R_0, \quad D(0) = D_0.
\end{equation}

This formulation remains well-suited for capturing epidemic-induced mortality, the primary concern in actuarial analysis during outbreaks. 
We are now ready to proceed with the well-posedness analysis of the simplified model in the next section. 


\subsection{Well-posedness Analysis}\label{S2.3}
This section focuses on analyzing the existence and uniqueness of a positive solution to the system~\eqref{E2.3}--\eqref{E2.4}. To verify the invariance and boundedness of the proposed model, we consider the total population
$$
  \ot{N} = \ot{S} + \ot{E} + \ot{I} + \ot{A} + \ot{R} + \ot{D} \stackrel{\eqref{E2.3}}{=} 0.
$$
This implies that the total population $N$ remains constant over time. Therefore, all solutions to \eqref{E2.3}--\eqref{E2.4} are bounded.
Similarly, summing the equations associated with the total living population $N_L$, yields
$$
  \ot{N_L} = -\delta_I I - \delta_A A.
$$
When $I, A \geq 0$, this implies $\ot{N_L} \leq 0$, meaning that $N_L$ decreases over time due to disease-induced mortality. 
Since $N_{L_0} = S_0 + E_0 + I_0 + A_0 + R_0 > 0$, and deaths occur continuously rather than instantaneously, it follows that $N_L > 0$ for all finite $t \geq 0$.
If, however $\ot{N_L} \geq 0$, which mathematically possible but realistic, $N_L$ would still increase over time. However, it would still remain strictly positive for all finite times due to the positivity of the initial condition.
\begin{proposition}\label{P1}
For any nonnegative initial conditions and strictly positive model parameters, the solutions to the system~\eqref{E2.3}--\eqref{E2.4} remain nonnegative for all $t \geq 0$.
\end{proposition}
\begin{proof}
First, we prove the positivity of the solution for each compartment to ensure the epidemiological relevance of the model. 
Consider the susceptible class $S$, which is governed by the following equation
\begin{equation*}
   \ot{S}= - \beta S \Bigl( \frac{I + \kappa A}{N_L} \Bigr).
\end{equation*}
We define the negative part of $S$ as $S^- = \max(0, -S)$. 
Multiplying both sides of the equation by $S^-$ and analyzing the resulting expression leads to 
\begin{equation*}
   \frac{1}{2} \ot{(S^-)^2} = \beta \Bigl( \frac{I(t) + \kappa A(t)}{N_L(t)} \Bigr) (S^-(t))^2.
\end{equation*}
Integrating this relation yields 
\begin{equation*}
   (S^-(t))^2 = (S^-_0)^2 \exp \biggl( \int_0^t 2\beta \,\frac{I(s) + \kappa A(s)}{N_L(s)}\, ds \biggr).
\end{equation*}
Since the initial condition satisfies $S_0 \geq 0$, we have $S^-_0 = 0$, implying $S^- = 0$ for all $t \geq 0$. Therefore, $S \geq 0$ for all $t \geq 0$. 

To prove that $E$ remains positive, we multiply its differential equation by its negative part, $E^-$.
Using the facts that $S\ge0$, $N_L > 0$, and that $S$, $E$, $I$, and $A$ are bounded, we obtain
\begin{equation*}
\frac{1}{2} \ot{}(E^-)^2 \leq -C (E^-)^2,
\end{equation*}
for some $C > 0$. Afterwards,
\begin{equation*}
   (E^-)^2 \leq (E^-_0)^2 \exp(-2Ct).
\end{equation*}
Since $E_0 \geq 0$, it follows that $E^-_0 = 0$. Therefore, $E \geq 0$ for all $t \geq 0$.

To demonstrate the nonnegativity of $I$, $A$, and $R$, $D$ two by two, consider the following
\begin{equation*}
  \ot{I}\Big|_{I=0} = \alpha p E \geq 0, \quad 
  \ot{A}\Big|_{A=0} = \alpha (1-p) E \geq 0,
\end{equation*}
and
\begin{equation*}
  \ot{R}\Big|_{R=0} = \gamma_I I + \gamma_A A \geq 0, \quad 
  \ot{D}\Big|_{D=0} = \delta_I I + \delta_A A \geq 0.
\end{equation*}
These derivatives indicate that the vector field on the boundaries of the positive orthant $\mathbb{R}_{+}^4$ points inward or is tangent to the boundary, i.e.\ the rates are nonnegative on the coordinate hyperplanes.
Therefore, any trajectory that starts in the interior of the positive orthant $\mathbb{R}_{+}^4$ remains in it for all future time. This ensures the positivity of the system's solutions.
\end{proof}

To prove the existence and uniqueness of solutions to the system of equations~\eqref{E2.3}--\eqref{E2.4}, we use standard results from the theory of ODEs, as supported by Proposition~\ref{P1}.
First, note that the right-hand side of the system is continuously differentiable and thus locally Lipschitz continuous in the positive orthant.
According to the Picard–Lindelöf theorem, this guarantees the local existence and uniqueness of solutions for given nonnegative initial conditions.
Furthermore, the positivity of the state variables $S$, $E$, $I$, $A$, $R,$ and $D$, together with the invariance of a bounded region, ensures that the solutions remain positive and bounded for all $t \geq 0$. This precludes finite-time blow-up and allows us to extend local solutions to global ones.
%
Therefore, the following corollary summarizes this subsection
\begin{corollary}\label{C1}
Given any positive initial data and strictly positive model parameters, system~\eqref{E2.3}--\eqref{E2.4} has a unique, global, positive, bounded solution for all $t\ge0$.
\end{corollary}

\section{Actuarial Applications}\label{S3}
The concept of insurance coverage for infectious diseases is similar to coverage for other risks, such as accidental death or property damage. However, it differs fundamentally from traditional property and casualty insurance because the number of policyholders paying premiums and the number of policyholders eligible for claims can fluctuate throughout an epidemic.

Since mortality analysis typically relies on ratios rather than absolute numbers, 
we define six deterministic functions $s(t)$, $e(t)$, $i(t)$, $a(t)$, $r(t)$, and $d(t)$,
that are obtained by
dividing equations~\eqref{E2.3} by the constant total population size $N$:
\begin{equation}\label{E3.1}
\left\{\begin{aligned}
\displaystyle \ot{s} &= - \beta \frac{N}{N_L} s (i + \kappa a), \\
\displaystyle \ot{e} &= \beta \frac{N}{N_L}  s (i + \kappa a) - \alpha e, \\
\displaystyle \ot{i} &= p\alpha e - (\gamma_I + \delta_I) i, \\
\displaystyle \ot{a} &= (1 - p)\alpha e - (\gamma_A + \delta_A) a, \\
\displaystyle \ot{r} &= \gamma_I i + \gamma_A a, \\
\displaystyle \ot{d} &= \delta_I i + \delta_A a,
\end{aligned}\right.
\end{equation}
supplied with the initial data
\begin{equation}\label{E3.2}
s(0) = s_0, \quad e(0) = e_0, \quad i(0) = i_0, \quad a(0) = a_0, \quad r(0) = r_0, \quad d(0) = d_0
\end{equation}
satisfying $s_0+e_0+i_0+a_0+r_0+d_0=1$. The components can be interpreted as the probabilities that a randomly selected individual belongs to one of the defined compartments at time $t$. 
However, the coupling terms introduce interdependence among the risks in the SEIARD model, which differs from the independent risk assumption typically found in multiple-decrement life insurance models.
With these probability functions defined, we can apply actuarial techniques to derive key quantities relevant to insurance coverage for infectious diseases. In the sequel we compute actuarial quantities as proposed by Feng and Garrido \cite{Feng2011} for the SEIARD model~\eqref{E3.1}.


\subsection{Premiums and Benefits}\label{S3.1}
We consider an infectious disease insurance plan in which premiums are continuously collected from individuals as long as they are susceptible. In other words, policyholders pay premiums in the form of continuous annuities for as long as they remain healthy and susceptible.
Conversely, the insurer continuously reimburses medical expenses for each infected individual throughout their treatment period. If an individual dies from the disease, coverage under the plan ends immediately.

Using International Actuarial Notation, we denote the \textit{actuarial present value} (APV) of premium payments over $t$ years by $\annuity{s+e}{t}$, where the superscript $s$ indicates that payments are made by susceptible and exposed individuals. 
Accordingly, the APV of benefit payments to infected individuals at a rate of one monetary unit per unit of time is denoted by $\annuity{s+e}{t}$.

To evaluate these annuities, we use the current payment approach. This method involves determining the present value of a payment due at time $t$, 
calculated as the discounted value of one monetary unit multiplied by the probability that the payment occurs.
We then integrate these present values over all possible times $t$. 
A comprehensive discussion of annuity evaluation techniques can be found in \cite{Bowers1997, Dickson2019}.

Accordingly, the insurer’s liability, the total discounted value of benefit payments over a $t$-year period, is expressed by the two infected classes as follows:
\begin{equation*}
    \annuity{i+a}{t} = \int_0^t \operatorname{e}^{-\delta \tau} \bigl(i(\tau)+a(\tau)\bigr) \,d\tau,
\end{equation*}
where $\delta\ge0$ represents the \textit{discounting interest rate}.
On the revenue side, the total discounted value of premium payments over the same $t$-year period is:
\begin{equation*}
    \annuity{s+e}{t}=\int_0^t \operatorname{e}^{-\delta \tau} \bigl(s(\tau)+e(\tau)\bigr)\,d\tau.
\end{equation*}

Our analysis is grounded in the \textit{fundamental principle of equivalence} for determining level premiums. According to this principle we have
\begin{equation*}
\mathbb{E}\bigl[\text{Present value of benefit outgo}\bigr] = \mathbb{E}\bigl[\text{Present value of premium income}\bigr].    
\end{equation*}
Based on this, the level premium for a unit annuity claim payment plan is given by:
\begin{equation*}
\bar{P}(\annuity{i+a}{t}) = \frac{\annuity{i+a}{t}}{\annuity{s+e}{t}},   
\end{equation*}
where $\annuity{i+a}{t}$ is the 
APV of benefit payments to the two infected classes $i(t)$ and $a(t)$ and $\annuity{s+e}{t}$ is the APV of premium payments from the susceptible and exposed classes.

For analytical tractability, we start with an insurance policy that has an infinite time horizon. When the policy duration is long, premiums computed under the assumption of an infinite term can reasonably approximate the true insurance cost.

\begin{proposition}
    Under the SEIARD model specified in equations~\eqref{E3.1}, assuming $\gamma_I= \gamma_A$, $\delta_I= \delta_A$, the following relation holds:
\begin{equation}\label{eq:prop1}
   \annuityinf{s+e} + \Bigl(1+\frac{\gamma_I + \delta_I}{\delta}\Bigr)\,\annuityinf{i+a}=
   \frac{1}{\delta}.
      \end{equation}
\end{proposition}
\begin{proof}
    Adding the first four equations in \eqref{E3.1} and integrating yields (with $s_0+e_0=1$)
    \begin{equation*}
    s(t)+e(t) -1 + i(t)+a(t) =  -(\gamma_I + \delta_I) \int_0^t i(\tau)\,d\tau
     -(\gamma_A + \delta_A) \int_0^t a(\tau)\,d\tau.
      \end{equation*}
     Next,  we assume $\gamma_I= \gamma_A$, $\delta_I= \delta_A$, we multiply by $\operatorname{e}^{-\delta t}$ and integrate from 0 to $\infty$:
     \begin{multline*}
   \int_0^\infty \operatorname{e}^{-\delta t} (s(t)+e(t))\,dt  + \int_0^\infty \operatorname{e}^{-\delta t} (i(t)+a(t)\,dt  -\frac{1}{\delta} \\
   =  - (\gamma_I+\delta_I)\int_0^\infty \operatorname{e}^{-\delta t}\int_0^t (i(\tau)+a(\tau))\,d\tau\,dt, 
      \end{multline*}
      Interchanging the order of integrals yields, 
       \begin{equation*}
       \begin{split}
   \annuityinf{s+e}
   +\annuityinf{i+a}
   -\frac{1}{\delta} &=  -  (\gamma_I+\delta_I) \int_0^\infty \operatorname{e}^{-\delta t} \int_0^t (i(\tau)+a(\tau))\,d\tau\,dt\\
   &=  -\frac{\gamma_I+\delta_I}{\delta}
   \int_0^\infty \operatorname{e}^{-\delta \tau}  (i(\tau)+a(\tau))\,d\tau
   =-\frac{\gamma_I+\delta_I}{\delta}\,\annuityinf{i+a},
    \end{split}
      \end{equation*}
      which proves~\eqref{eq:prop1}.
\end{proof}

Note that the right-hand side of \eqref{eq:prop1} represents the present value of a unit perpetual annuity at a constant interest rate $\delta$.
The left-hand side captures the distribution of payments across compartments, namely
(i) the APV of continuous premium payments by susceptibles or exposed is given by $\annuityinf{s+e}$,
(ii) the APV of continuous benefit payments to infected 
individuals is $\annuityinf{i+a}$,
  (iii) the APV of the perpetuities granted to individuals transitioning from class $i$ or $a$ to class $r$ or $d$ is $((\gamma_I+\delta_I)/\delta)\annuityinf{s+e}$, where each individual in class $r$ is entitled to a perpetuity worth $1/\delta$ at the time of transition.
Thus, the total APV of all six components equals that of a unit perpetuity, ensuring fairness and consistency under the equivalence principle.

Using equation~\eqref{eq:prop1}, we can derive the net level premium for an infinite-term policy, where premiums and claims are structured as continuous annuity payments. The formula is given by
\begin{equation*}
\bar{P}(\annuityinf{i+a}) = \frac{\annuityinf{i+a}}{\annuityinf{s+e}} 
=\frac{\delta \annuityinf{i+a}}{1-(\delta+\gamma_I+\delta_I)\annuityinf{s+e}}.
\end{equation*}



\subsection{Reserve Calculation}\label{S3.2}
In the context of epidemic-linked insurance products, a key actuarial quantity is the prospective reserve. This is defined as the insurer's net liability at time $t$, taking into account the expected present value of future benefits and premiums.
The SEIARD model structure natu\-rally supports this formulation by identifying subpopulations that are either responsible for generating claims (infectious classes $i$ and $a$, as well as deaths in $d$) or contributing premiums (susceptible $s$, exposed $e$, and recovered $r$).
Let $\pi$ be the premium rate per person per unit time, $b_I$ and $b_A$ be the per-unit-time benefits paid to symptomatic and asymptomatic infectives, respectively, $L_D$ be the lump-sum death benefit paid upon death, and $\delta$ be the discount rate.
We define the prospective reserve at time $t$, denoted $V$, as
$$
   V(t) = \text{APV of future benefits to } i, a, d - \text{APV of future premiums from } s, e, r.
$$
Thus, the reserve at time $t$, is given by 
\begin{equation}
\begin{aligned}\label{eq:reserve}
   V(t) &= \int_t^T e^{-\delta (x - t)} \bigl[ b_I i(x) + b_A a(x) + L_D ( \delta_I i(x) + \delta_A a(x) ) \bigr] \, dx \\
   &\qquad - \pi \int_t^T e^{-\delta (x - t)} \bigl[ s(x) + e(x) + r(x) \bigr] \, dx.
\end{aligned}
\end{equation}
This expression~\eqref{eq:reserve} captures the evolving nature of the epidemic, enabling a dynamic assessment of insurer liabilities.
For example, a sharp increase in $i(t)$ or $\delta_I$ leads to an increase in expected claims and, consequently, a higher reserve requirement.
Conversely, including $r(t)$ as a premium-paying class stabilizes income since recovered individuals are assumed to maintain coverage.
By analyzing $V$ under various epidemic scenarios, insurers can evaluate solvency margins, assess the robustness of pricing strategies, and design more resilient products.

In actuarial models of epidemic insurance, the \textit{premium rate} $\pi$  plays a central role in determining the financial soundness of the insurance scheme. 
It represents the continuous income collected from insured individuals to finance future claim payments. 
From a modeling perspective, $\pi$ is a free parameter. 
However, for the reserve function $V$ to remain meaningful in actuarial terms, $\pi$ must be strictly positive. 
A negative or zero premium rate implies that the insurer is offering coverage for free or subsidizing the insured, which is unrealistic. Furthermore, the premium rate must be carefully calibrated to ensure that the reserve function remains nonnegative throughout the insurance horizon $[0, T]$, reflecting the requirement that the insurer remains solvent at all times.\\
If $\pi$ is set too low, the reserve $V$ may become negative during the simulation period. This indicates that the scheme is underfunded and incapable of covering its expected liabilities.
Conversely, if $\pi$ is set too high, the reserve may exceed the expected claims by a large margin, reflecting overpricing or inefficiency.\\
To determine an admissible value of $\pi$ for a $T$-day insurance policy, the rate should be chosen such that $V(t) \geq 0$ for all $t \in [0, T]$. This leads to the actuarial constraint
\begin{equation*}
0 < \pi \leq \pi^* := \min_{t \in [0, T]} \frac{\int_t^T e^{-\delta (x - t)} \bigl[ b_I i(x) + b_A a(x) + L_D ( \delta_I i(x) + \delta_A a(x) ) \bigr]\,dx}{\int_t^T e^{-\delta (x - t)} \bigl[ s(x) + e(x) + r(x) \bigr]\,dx}.
\end{equation*}
This inequality ensures that the present value of future premium income never falls below the present value of future liabilities.
It also establishes a practical upper bound for $\pi$, which can be approximated numerically based on model trajectories.

\subsection{Force of Infection, Mortality, and Removal}\label{S3.3}
The concept of the force of infection plays a pivotal role in linking epidemic dynamics with actuarial modeling. Within the SEIARD framework, we consider two formulations of this quantity that offer complementary perspectives. 
The first is the mechanistic force of infection, denoted $\lambda(t)$, which captures the transmission rate from infectious individuals to susceptibles based on model parameters and compartment sizes. It is defined by
\begin{equation}\label{ForInfII}
  \lambda(t) = \beta \frac{i(t) + \kappa a(t)}{N_L(t)},
\end{equation}
and represents the instantaneous hazard rate at which susceptibles acquire infection.
In contrast, we also define an empirical force of infection, denoted $\mu^{s+e}_t$, which describes the observed rate of depletion of the combined susceptible and exposed classes
\begin{equation}\label{ForInfI}
\mu^{s+e}_t = -\frac{s'(t)+e'(t)}{s(t)+e(t)}, \quad t\ge0.
\end{equation}
While $\lambda(t)$ arises from the model's transmission mechanisms, $\mu^{s+e}_t$ can be seen as an effective or realized force of infection, incorporating all processes contributing to the decline of $s(t) + e(t)$, including new infections and progression to infectious states. Comparing both expressions allows us to validate the model’s internal consistency and examine differences between theoretical assumptions and actual epidemic dynamics.

From an actuarial standpoint, the mechanistic force $\lambda(t)$ governs the flow of individuals from $s(t)$ to $e(t)$, which in turn influences future claims and premium base dynamics. For instance, higher values of $\lambda(t)$ lead to increased exposure and eventually higher numbers in the $i(t)$ and $a(t)$ compartments, thus amplifying benefit payouts.
Moreover, $\lambda(t)$ allows the computation of epidemic-adjusted survival probabilities for individuals initially in the susceptible class. The probability that an individual avoids infection up to time $t$ is given by
\begin{equation}\label{ProbForInfII}
  p_s(t) = \exp\Bigl( - \int_0^t \lambda(x)\, dx \Bigr),
\end{equation}
which can be used in pricing life or health insurance contracts under epidemic conditions. The behavior of $\lambda(t)$ over time can also inform public health-linked underwriting policies, such as temporary exclusions or adjusted premiums during high-transmission periods.
Thus, incorporating the force of infection provides a rigorous basis for modeling the impact of disease transmission on claim frequency and insurance fund volatility.

In parallel with the force of infection, we distinguish between two complementary notions capturing the impact of disease-related mortality in the SEIARD system~\eqref{E3.1}. The first is the force of mortality, denoted $\mu^d(t)$, which quantifies the instantaneous death rate from both symptomatic and asymptomatic infections. It is given by the mechanistic expression
\begin{equation}\label{ForMort}
   \mu^d(t) = \frac{\delta_I i(t) + \delta_A a(t)}{N_L(t)},
\end{equation}
and governs the flow of individuals from the living population to the deceased class $d(t)$. This mortality force directly influences the severity of insurance claims and acts as a dynamic hazard function that adjusts classical survival models to epidemic conditions. The epidemic-adjusted survival probability is then defined as
\begin{equation}\label{ProbForMort}
   p_l(t) = \exp\Bigl( - \int_0^t \mu^d(x)\, dx \Bigr),
\end{equation}
representing the likelihood that an individual survives the epidemic up to time $t$. This formulation is essential for pricing annuities, valuing contingent liabilities, and computing reserves under mortality shocks.

Complementing this, we introduce the empirical force of removal, denoted $\mu^{i+a}_t$, which captures the effective rate at which individuals exit the infectious classes due to recovery or death
\begin{equation}\label{ForRemov}
\mu^{i+a}_t = -\frac{i'(t)+a'(t)}{i(t)+a(t)}, \quad t\ge0.
\end{equation}
From the system dynamics, this expression expands to
\begin{equation*}
\mu^{i+a}_t = \frac{(\gamma_I + \delta_I) i(t) + (\gamma_A + \delta_A) a(t) - \alpha e(t)}{i(t) + a(t)},
\end{equation*}
showing how removals result from both mortality and recovery processes, as well as the progression from the exposed state. While $\mu^d(t)$ isolates disease-induced mortality, $\mu^{i+a}_t$ accounts for the total outflow from the infectious compartments, thus providing an empirical benchmark against which mechanistic assumptions can be compared.

More generally, these force-based formulations yield integral representations of compartment sizes. For instance, the susceptible class satisfies
\begin{equation*}
   s(t)=s_0\exp\Bigl\{ -\int_0^t \mu^s_\tau\,d\tau\Bigr\}
   =s_0\exp\Bigl\{ -\beta \int_0^t \frac{N}{N_L(\tau)}(i(\tau) + \kappa a(\tau)) \,d\tau\Bigr\},
\end{equation*}
and analogous relations hold for the other classes. This highlights the structural role played by these hazard rates in shaping population trajectories, and offers a natural bridge to actuarial measures such as survival probabilities and expected times to death or recovery.

\section{Numerical Scheme}\label{S4}
We employ the nonstandard finite difference (NSFD) method introduced by Mickens \cite{Mickens1993} to approximate the SEIARD model~\eqref{E3.1} numerically.
This method is particularly suitable for epidemic models because it is designed to preserve the continuous system's key qualitative properties, such as positivity, boundedness, and dynamic consistency.
These properties are often violated by standard numerical schemes, such as the explicit Euler method or the classical Runge-Kutta methods \cite{Zinihi2025NSFD, Ehrhardt2013, Costa2024, Maamar2023}.

Although standard schemes are widely used, they can produce nonphysical artifacts, such as negative compartment sizes or spurious equilibria, especially when applied to nonlinear systems with stiff dynamics or conservation constraints.
These issues are particularly problematic in epidemiological and actuarial settings where state variables (e.g., infection ratios and mortality proportions) must be interpreted in a biologically and financially meaningful way.
For instance, as illustrated in \cite{Zinihi2025NSFD}, even in reaction-diffusion settings, the traditional finite difference method may result in negative population sizes for exposed individuals under specific conditions. 
In contrast, the NSFD method consistently maintains positivity and qualitative realism.

In the NSFD framework, the time derivative $\dot{u}(t)$ of a generic state variable $u$ is approximated using a nontrivial denominator function $\varphi(k)$, where $k=\Delta t$ is the time step
\begin{equation*}
  \ot{u}\bigg|_{t=t_n} \approx \frac{u^{n+1} - u^n}{\varphi(k)},
\end{equation*}
with $\varphi(k) > 0$ and $\varphi(k) = k + \mathcal{O}(k^2)$. One common choice that guarantees positivity preservation is $\varphi(k) = \frac{e^{\mu k} - 1}{\mu}$, where $\mu$ is the natural mortality rate, as discussed in \cite{Zinihi2025NSFD}.

Let $k$ be the time step size, and denote $s^n \approx s(t_n)$, $e^n \approx e(t_n)$, and so on. 
The NSFD discretization of the SEIARD model~\eqref{E3.1} becomes
\begin{equation}\label{E2.5scheme} 
\left\{\begin{aligned}
\displaystyle \frac{s^{n+1} - s^n}{\varphi(k)} &= - \beta \frac{N}{N_L^n} s^{n+1} (i^n + \kappa a^n), \\
\displaystyle \frac{e^{n+1} - e^n}{\varphi(k)} &= \beta \frac{N}{N_L^n} s^{n+1} (i^n + \kappa a^n) - \alpha e^{n+1}, \\
\displaystyle \frac{i^{n+1} - i^n}{\varphi(k)} &= p\alpha e^{n+1} - (\gamma_I + \delta_I) i^{n+1}, \\
\displaystyle \frac{a^{n+1} - a^n}{\varphi(k)} &= (1 - p)\alpha e^{n+1} - (\gamma_A + \delta_A) a^{n+1}, \\
\displaystyle \frac{r^{n+1} - r^n}{\varphi(k)} &= \gamma_I i^{n+1} + \gamma_A a^{n+1}, \\
\displaystyle \frac{d^{n+1} - d^n}{\varphi(k)} &= \delta_I i^{n+1} + \delta_A a^{n+1}.
\end{aligned}\right.
\end{equation}
This scheme follows Mickens’s rules \cite{Mickens1993};
\begin{itemize}
\item Nonlinear terms such as $s(i + \kappa a)$ are discretized non-locally, i.e., using a mix of time levels (e.g., $s^{n+1}(i^n + \kappa a^n)$) to maintain positivity;

\item The discrete derivative uses a nonlinear denominator function $\varphi(k)$ that reflects the asymptotic behavior of the system;

\item The scheme is explicitly solvable in a sequential manner, with each variable updated in order.
\end{itemize}
Afterwords,
\begin{equation}\label{NSFD}
\left\{\begin{aligned}
s^{n+1} &= \dfrac{s^n}{1 + \beta \frac{N}{N_L^n} (i^n + \kappa a^n) \varphi(k)}, \\
e^{n+1} &= \dfrac{e^n  + \beta \frac{N}{N_L^n} s^{n+1} (i^n + \kappa a^n) \varphi(k)}{1 + \alpha\varphi(k)}, \\
i^{n+1} &= \dfrac{i^n + p \alpha e^{n+1} \varphi(k) }{1 + (\gamma_I + \delta_I)\varphi(k)}, \\
a^{n+1} &= \dfrac{a^n + \alpha (1 - p) e^{n+1} \varphi(k) }{1 + (\gamma_A + \delta_A) \varphi(k)}, \\
r^{n+1} &= r^n + \varphi(k) (\gamma_I i^{n+1} + \gamma_A a^{n+1}), \\
d^{n+1} &= d^n + \varphi(k) (\delta_I i^{n+1} + \delta_A a^{n+1}).
\end{aligned}\right.
\end{equation}

Let us briefly comment on the discretization of the nonlinear incidence terms, which are bilinear in nature. In particular, the infection term $\beta \tfrac{N}{N_L(t)} s(t) (i(t) + \kappa a(t))$ is discretized in the first two equations of \eqref{NSFD} as $\beta \tfrac{N}{N_L^n} s^{n+1} (i^n + \kappa a^n)$, rather than $\beta \tfrac{N}{N_L^n} s^n (i^n + \kappa a^n)$ or $\beta \tfrac{N}{N_L^{n+1}} s^{n+1} (i^{n+1} + \kappa a^{n+1})$. 
The guiding principle here is to evaluate exactly one factor at the new time level, specifically, the variable whose time derivative appears in the equation (in this case, $s$ or $e$). 
This semi-implicit treatment ensures the positivity of the numerical solution while maintaining the explicit solvability of the scheme. 
This strategy aligns with the general rules of NSFD methods and has been successfully applied in similar epidemiological contexts (see e.g., \cite{Ehrhardt2013, Mickens1993, Zinihi2025NSFD}).

By construction, this NSFD scheme ensures that all compartments $s^n, e^n, i^n, a^n, r^n$, and $d^n$ remain non-negative and the discrete total normalized population 
$s^n+e^n+i^n+a^n+r^n+d^n$
remains constant at 1 (simply add all equations in \eqref{E2.5scheme}), provided that the initial data satisfy this normalization condition. 
These qualitative properties of the NSFD scheme \eqref{NSFD} also ensure its stability.
In summary, the NSFD method provides a structure-preserving and dynamically consistent approach to numerically solving epidemiological models such as \eqref{E3.1}. 
Its robustness, especially in maintaining positivity and avoiding nonphysical behaviors, makes it a compelling choice for actuarial analyses and public health policy modeling.

\section{Numerical Results}\label{S5}
In this section, we present numerical results that illustrate the dynamic behavior of the SEIARD model~\eqref{E3.1}--\eqref{E3.2} using the NSFD scheme~\eqref{NSFD}. 
Our goal is to simulate the temporal evolution of the epidemiological compartments under realistic conditions and evaluate the model’s capacity to generate plausible outbreak scenarios.
The analysis is divided into two parts. First, we describe the parameter values and initial conditions used in the simulations based on empirical data from a real-world SARS-CoV-2 outbreak. Second, we apply the proposed NSFD scheme to explore population dynamics.
Additionally, we examine the evolution of the force of infection, mortality, and the actuarial reserve, which are key components of the model's financial interpretation.
The simulations confirm the consistency, positivity, and stability of the numerical scheme and demonstrate its relevance to epidemiological modeling and actuarial risk assessment.

\subsection{Parameter Estimate and Initial Conditions}\label{S5.1}
We calibrate the model parameters and initial conditions of the SEIARD model~\eqref{E3.1}--\eqref{E3.2} based on data from the early phase of the 2020 Mexican SARS-CoV-2 epidemic \cite{Github2020, Conacyt2020} to perform realistic numerical simulations.
Specifically, we refer to the study \cite{AvilaPoncedeLen2020}, which fitted an SEIARD-type model to \textit{real case counts} from March to July 2020.
These estimates capture essential epidemiological characteristics of the disease dynamics during the initial wave and are consistent with values reported in the literature across several settings.
The following Table~\ref{Tab2} summarizes the parameter values used in the numerical scheme~\eqref{NSFD}.

\begin{table}[H]
\centering
\setlength{\tabcolsep}{0.8cm}
\caption{Estimated parameter values based on COVID-19 data from Mexico \cite{Github2020, Conacyt2020}.}\label{Tab2}
\adjustbox{max width=\textwidth}{
\begin{tabular}{cc||cc}
\hline 
\textbf{Parameter} & \textbf{Value} & \textbf{Parameter} & \textbf{Value}\\
\hline \hline 
$\beta$ & 0.3 &  $\gamma_I$ & 0.2 \\
\hline
$\kappa$ & 0.7 & $\delta_I$ & 0.007 \\
\hline
$\alpha$ & 0.192 & $\gamma_A$ & 0.1 \\
\hline
$p$ & 0.14 & $\delta_A$ & 0.001 \\
\hline
\end{tabular}
}
\end{table}

These values reflect a disease with a relatively high transmission rate, a moderate incubation period, and a progression that differs between symptomatic and asymptomatic cases.
Notably, only a small fraction of exposed individuals become symptomatic ($p = 0.14$), while the rest remain asymptomatic, which is consistent with early clinical findings from the surveillance of the SARS-CoV-2 virus.

For the initial conditions, we assume the following normalized values representing the early stage of an outbreak in a population of size $N$, where each compartment is expressed as a proportion of the total population
\begin{equation*}
   s(0) = 0.9999, \ e(0) = 0.00005, \ i(0) = 0.00003, \ a(0) = 0.00002, \ r(0) = 0, \ d(0) = 0.
\end{equation*}

These values correspond to a scenario in which nearly the entire population is susceptible and there are only a few initial cases that are either exposed, symptomatic, or asymptomatic.
This setup is ideal for analyzing the initial spread and early intervention strategies, particularly in the context of epidemic insurance planning or actuarial risk evaluation.

Unless stated otherwise, we assume that the total living population $N_L$ is approximately equal to $N$, and that there is negligible demographic turnover during the short-to-medium simulation horizon 
(e.g., 200 days in \cite{AvilaPoncedeLen2020}).
The time step $k$ used in the NSFD scheme~\eqref{NSFD} is set to $k = 1$ day, 
and the denominator function is chosen as $\varphi(k) = \frac{e^{\mu k} - 1}{\mu}$,
with $\mu \to 0$ (i.e., no natural mortality), simplifying to $\varphi(k) = k$.

\subsection{Simulations}\label{S5.2}
Figure~\ref{F3} shows the 
evolution of each compartment in the SEIARD epidemiological model, 
capturing the 
progression of the outbreak. The proportion of susceptible individuals, $s(t)$, declines sharply once the epidemic begins and eventually stabilizes at a lower level, reflecting significant population-level exposure to the disease. 
The exposed class $e(t)$ displays a typical rise-and-fall profile, peaking around day 150 as individuals transition to the infectious stages. 
The symptomatic infected class $i(t)$ exhibits a relatively low peak, suggesting either reduced transmission from symptomatic individuals or more rapid recovery or mortality.
In contrast, the asymptomatic infected class $a(t)$ reaches a much higher peak and dominates the transmission dynamics.
The recovered population $r(t)$ increases steadily and ultimately becomes the largest compartment, indicating broad population recovery or immunity. Meanwhile, the deceased class $d(t)$ grows gradually and levels off, representing the cumulative impact of mortality.
\begin{figure}[H]
\centering
\includegraphics[width=0.96\textwidth]{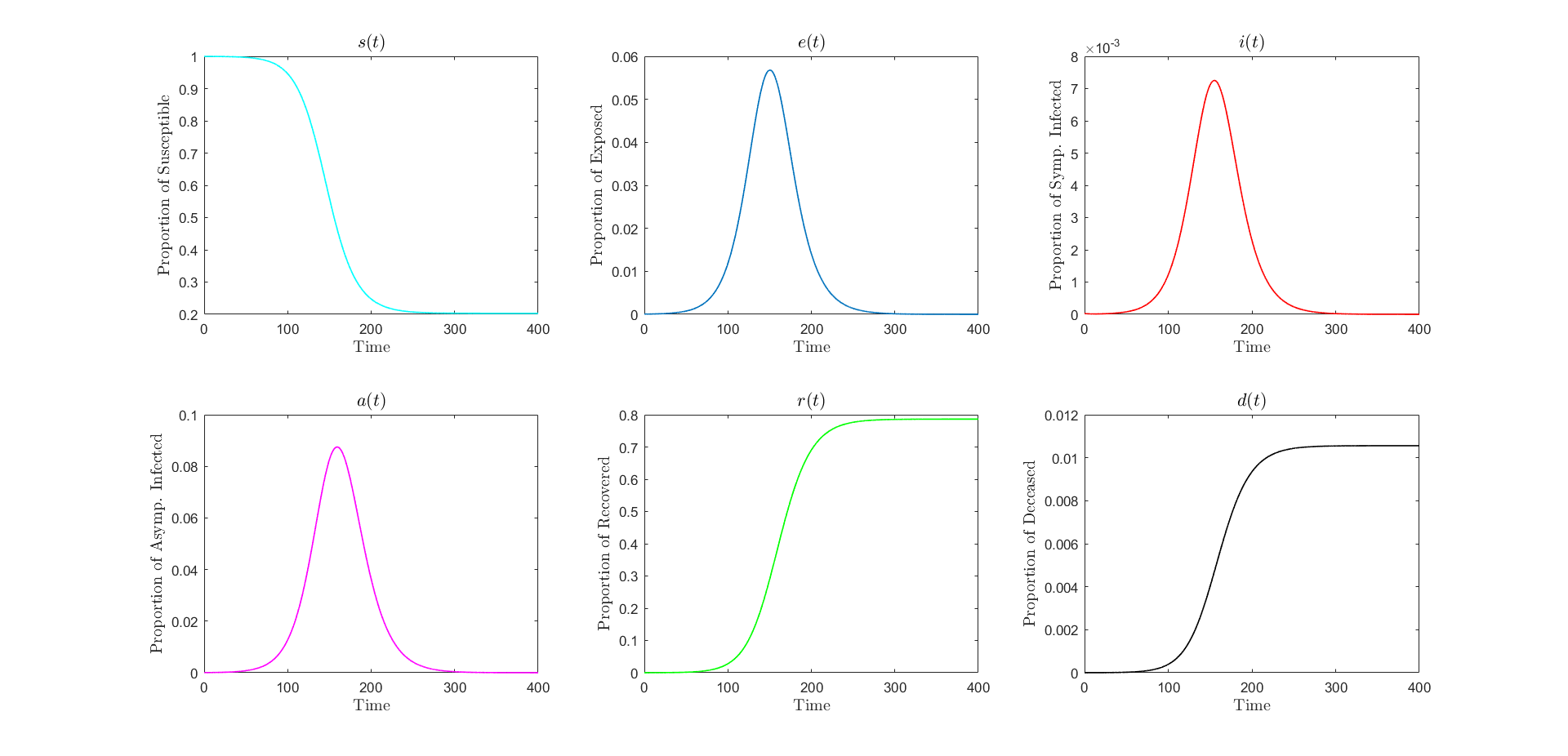}
\caption{Individual dynamics of the SEIARD compartments~\eqref{E3.1}--\eqref{E3.2}.}\label{F3}
\end{figure}

Figure~\ref{F4} illustrates the compartmental dynamics relevant to the financial aspects of epidemic insurance modeling. The left panel displays the time evolution of the premium-paying compartments $s(t)$, $e(t)$, and $r(t)$, and are therefore assumed to contribute regular insurance premiums. The right panel presents the dynamics of the benefit-related compartments $i(t)$, $a(t)$, and $d(t)$. 
These compartments are crucial for estimating insurance liabilities arising from hospitalization costs and death benefits. Together, this figure complements the epidemiological view by highlighting the population groups involved in premium inflows and benefit outflows, i.e. Figure~\ref{F2}.
\begin{figure}[H]
\centering
\includegraphics[width=1\textwidth]{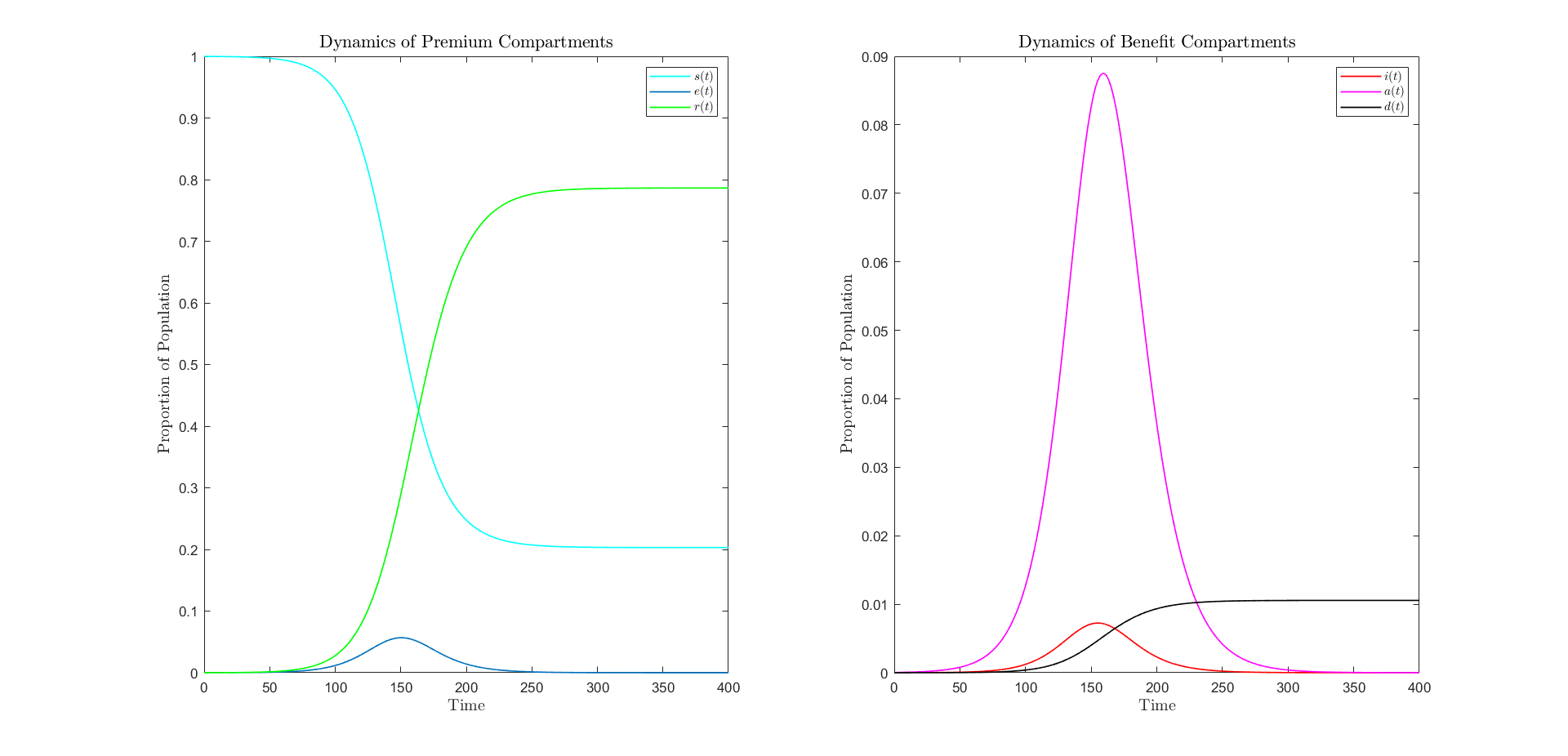}
\caption{Dynamics of premium-paying (\(s(t), e(t), r(t)\)) and benefit-receiving (\(i(t), a(t), d(t)\)) compartments.}\label{F4}
\end{figure}

Figure~\ref{F5} presents the temporal dynamics of the force of infection $\mu_t^{s+e}$ and the force of removal $\mu_t^{i+a}$, which are key elements in modeling epidemic-induced insurance transitions. 
The left panel shows the evolution of the infection force $\mu_t^{s+e}$, which peaks sharply around day 160 before rapidly declining, indicating the period of highest transmission intensity.
The middle panel shows the force of removal $\mu_t^{i+a}$, starting from a negative value and gradually increasing towards a plateau, reflecting the accumulation of recovery and mortality effects over time. 
The right panel combines both curves for direct comparison, emphasizing the temporal lag between peak infection pressure and the stabilization of removal effects. 
This figure highlights the dynamic interplay between infection spread and population outflow from the infectious compartments, essential for assessing risk exposure and estimating insurance liabilities throughout the epidemic lifecycle.

\begin{figure}[htb]
\centering
\includegraphics[width=1\textwidth]{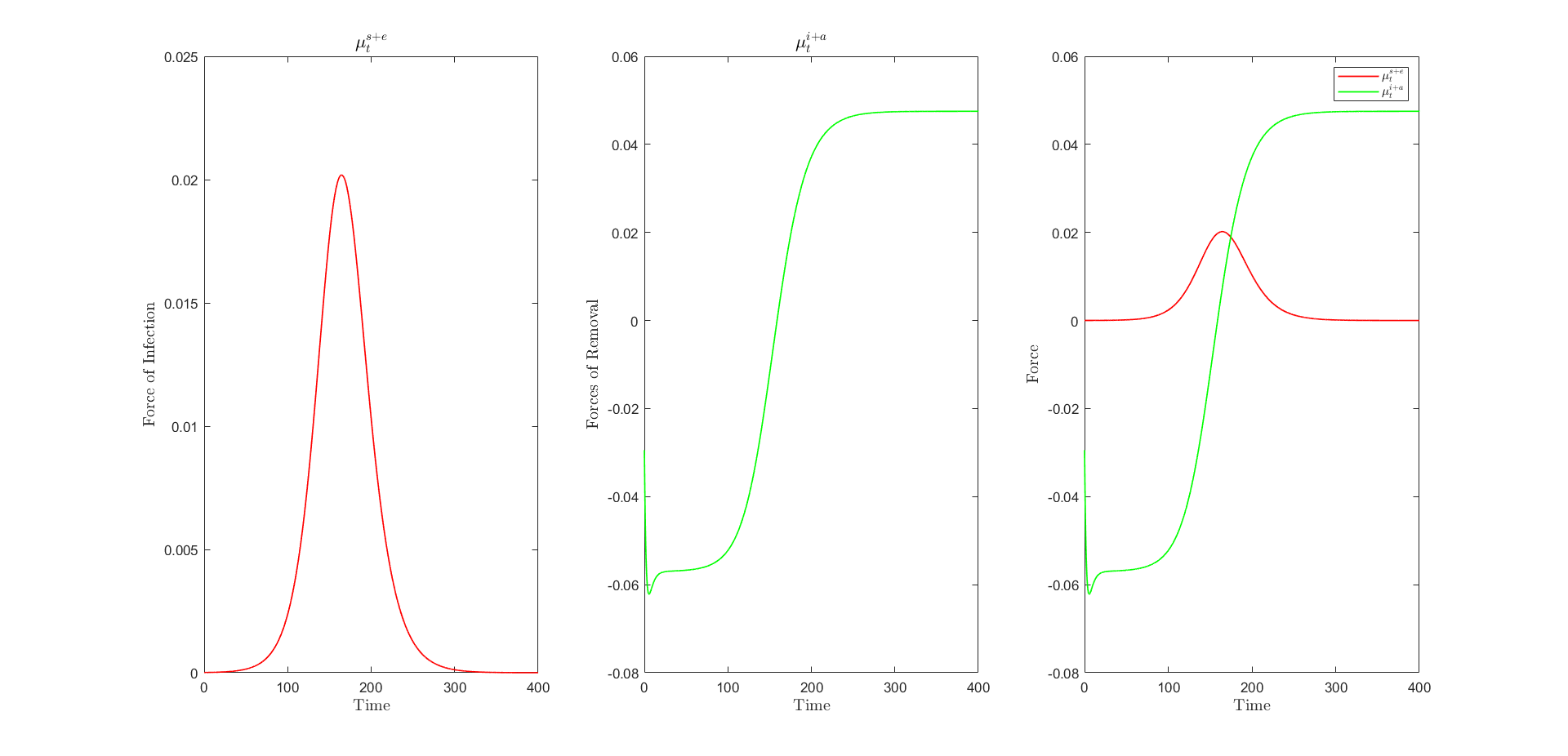}
\caption{Temporal evolution of the force of infection~\eqref{ForInfI} and force of removal~\eqref{ForRemov}.}\label{F5}
\end{figure}

Figure~\ref{F6} presents a comprehensive view of the infection and mortality risks during the epidemic, as well as their consequences on survival probabilities.
The top row shows the evolution of the force of infection $\lambda(t)$ (left panel), the force of mortality $\mu^d(t)$ (middle panel), and their direct comparison (right panel). The force of infection peaks around day 160, corresponding to the period of maximum transmission intensity. 
Its bell-shaped curve reflects the typical trajectory of an epidemic wave. 
In contrast, the force of mortality remains several orders of magnitude smaller but follows a similar temporal pattern, indicating that death occurrences are closely linked to the infectious burden, albeit with much lower magnitude.
The bottom row displays the evolution of the probabilities $p_s(t)$ (left) and $p_l(t)$ (middle), and their combined view (right). The infection-free probability $p_s(t)$ decreases rapidly from 1 to approximately 0.23, highlighting the high cumulative exposure of the population to the infection. On the other hand, the epidemic-adjusted survival probability $p_l(t)$ remains close to 1, with a modest decline due to the relatively low mortality force. This gap between $p_s(t)$ and $p_l(t)$ illustrates the fact that while a large proportion of the population may be infected, only a small fraction succumb to the disease under the assumed parameter regime.
These quantities are essential for quantifying the individual-level risks and for evaluating insurance models that account for both infection and mortality events. 
In particular, the joint dynamics of $\lambda(t)$, $p_s(t)$, $\mu^d(t)$, and $p_l(t)$ inform the temporal evolution of insurance exposure and help calibrate risk-adjusted actuarial strategies.
\begin{figure}[H]
\centering
\includegraphics[width=1\textwidth]{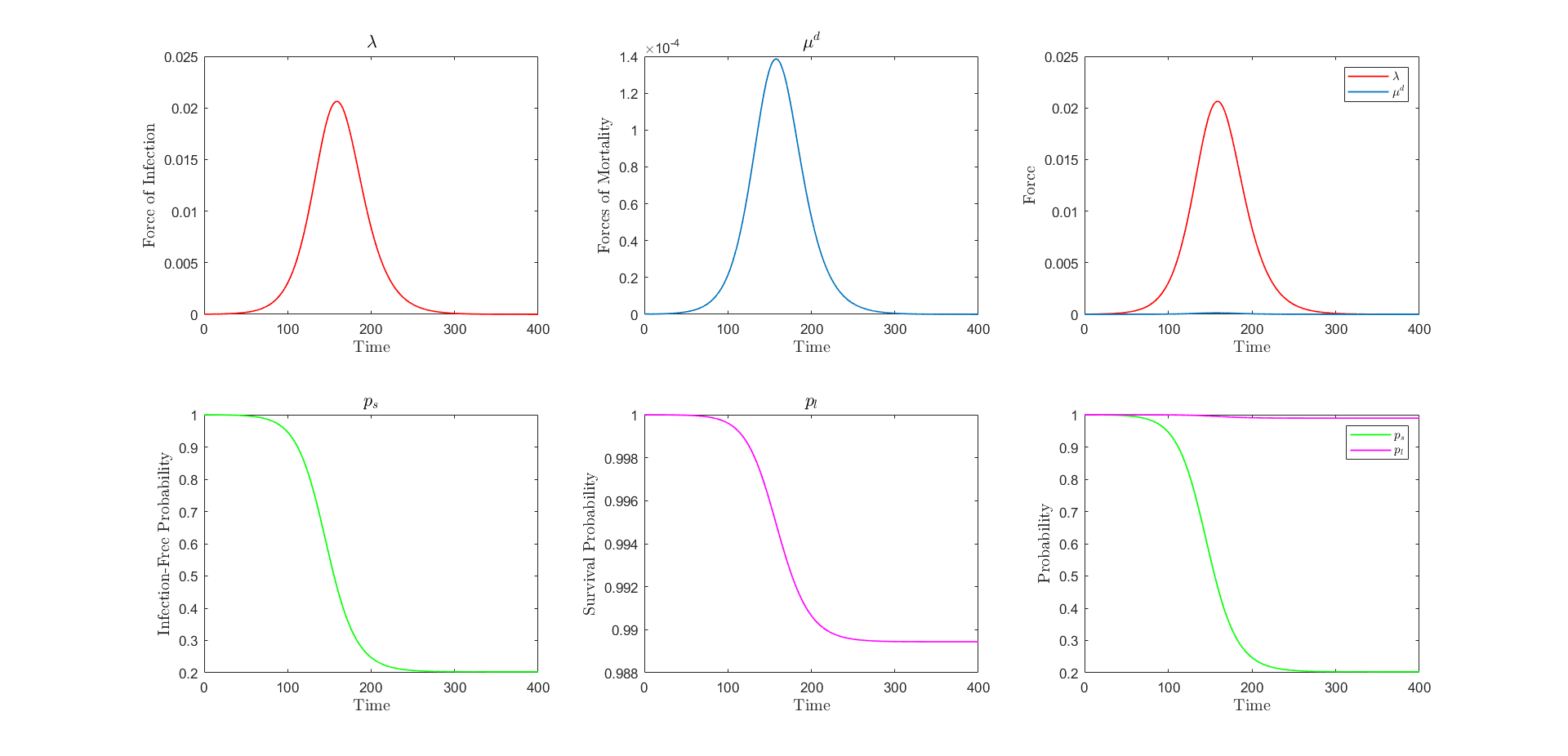}
\caption{Temporal dynamics of the force of infection~\eqref{ForInfII}, the force of mortality~\eqref{ForMort}, the infection-free probability~\eqref{ProbForInfII}, and the epidemic-adjusted survival probability~\eqref{ProbForMort}.}\label{F6}
\end{figure}

Figure~\ref{F7} illustrates the comparison between two expressions of the force of infection: the empirical formulation $\mu^{s+e}_t$, and the mechanistic expression $\lambda(t)$. 
Both curves display a similar bell-shaped profile, peaking near day 170, which corresponds to the period of highest epidemic intensity. 
The empirical force $\mu^{s+e}_t$, shown in dark red, captures the relative rate of depletion in the non-infected population (susceptible and exposed), while the mechanistic force $\lambda(t)$, shown in dark blue, quantifies the instantaneous risk of infection based on the infectious and asymptomatic classes. 
Although both curves are closely aligned in shape and amplitude, a slight phase shift is observed, with $\mu^{s+e}_t$ peaking marginally earlier. This distinction reflects the dynamic lag between exposure and infectiousness. 
Overall, the consistency between the two formulations validates the model’s internal coherence and provides complementary perspectives: $\lambda(t)$ is appropriate for modeling transmission processes, whereas $\mu^{s+e}_t$ offers an aggregate view of epidemic pressure from the perspective of the non-infected population.
\begin{figure}[H]
\centering
\includegraphics[width=1\textwidth]{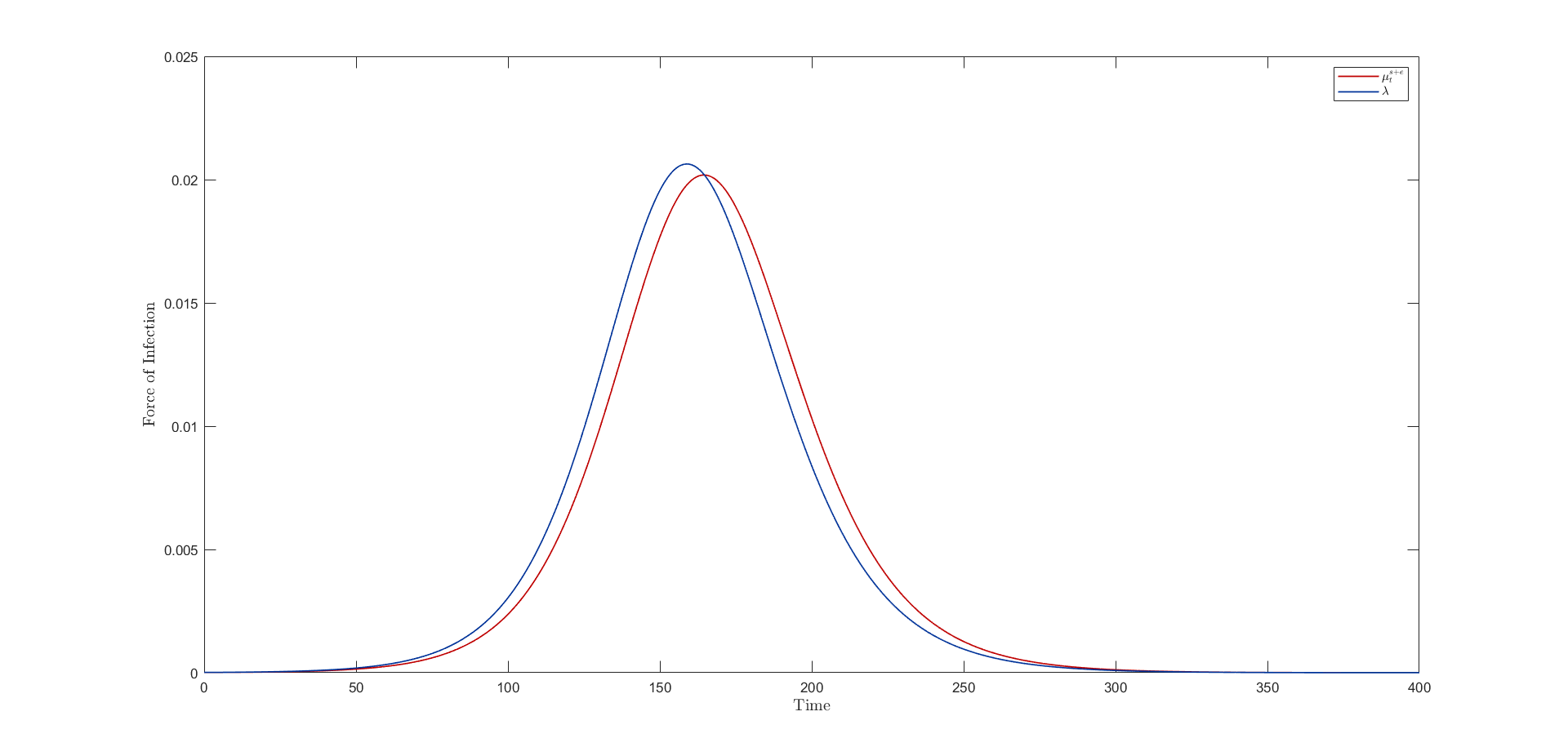}
\caption{Comparison of the two forces of infection~\eqref{ForInfII} and \eqref{ForInfI}.}\label{F7}
\end{figure}

Figure~\ref{F8} illustrates the time-dependent behavior of the reserve function $V(t)$, defined by the net present value of the future expected benefits minus the expected liabilities, over the time horizon $[t, T]$. The reserve accounts for the insured benefits paid in the event of infection or disease-related death and subtracts the collected premiums from the susceptible, exposed, and recovered populations.\\
Numerically, the o\textit{ptimal admissible premium rate} is $\pi^* = 2.9861\,e^{-06}$. The figure illustrates the evolution of the reserve function $V(t)$ under three different premium rates: $\pi = \pi^* - \varepsilon$ (top left), $\pi = \pi^*$ (top right), and $\pi = \pi^* + \varepsilon$ (bottom left), where $\varepsilon$ is a small positive perturbation. These simulations allow us to assess the sensitivity of the reserve to variations in the premium rate and validate the actuarial balance at $\pi = \pi^*$.
In each scenario, the reserve increases initially as the epidemic unfolds and reaches a peak near day 150, which coincides with the peak in infection prevalence. Then, it declines as the epidemic wanes and the liabilities diminish.
The comparison plot (bottom right) shows how the premium rate affects the reserve trajectory:
\begin{itemize}
\item For $\pi = \pi^* - \varepsilon$, the reserve reaches the highest value but eventually declines significantly, even becoming negative near the end of the period, indicating a potential deficit.

\item For $\pi = \pi^*$, the reserve stays strictly positive and tends to zero as $t \to T$, confirming that the premium exactly covers the expected liabilities.

\item For $\pi = \pi^* + \varepsilon$, the reserve is lower overall, with an early peak followed by a gradual decay and a positive terminal value, reflecting a surplus due to excess premium.
\end{itemize}
This analysis confirms the sensitivity of the reserve function to small deviations in the premium rate. It highlights the importance of carefully calibrating $\pi$ to avoid underfunding or unnecessary surplus accumulation. In particular, $\pi^*$ ensures that the insurer remains solvent with zero expected surplus or deficit at the terminal time $T$.
\begin{figure}[H]
\centering
\includegraphics[width=1\textwidth]{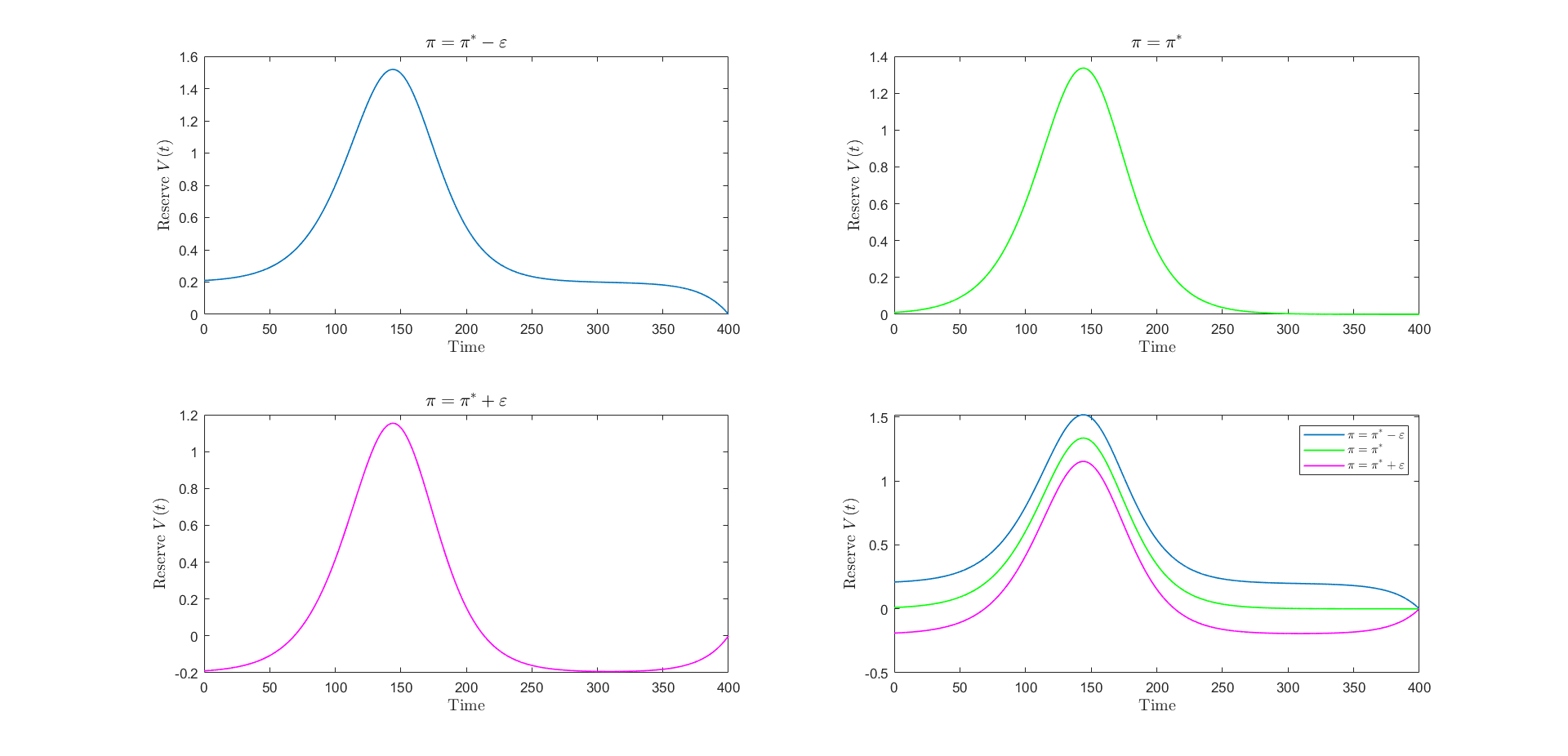}
\caption{Temporal evolution of the reserve function~\eqref{eq:reserve} for three different premium rates.}\label{F8}
\end{figure}

\subsection{Actuarial Interpretation of the Reserve Function}\label{S5.3}
In actuarial terms, the reserve function $V(t)$ quantifies the net liability of the insurer at time $t$. It serves as a key indicator of the insurer’s financial balance throughout the coverage period;

\begin{itemize}
\item If $V(t) > 0$: The insurer holds a positive reserve at time $t$, implying that the expected liabilities outweigh the incoming premiums. This situation requires additional capital to ensure solvency and reflects a conservative pricing strategy.

\item If $V(t) < 0$: The insurer has a negative reserve, indicating that the present value of premium income exceeds expected future costs. While this may suggest a surplus, it typically signals an overpriced premium, raising concerns about fairness, regulatory compliance, and long-term sustainability. Persistent negative reserves are generally unacceptable in actuarial practice.

\item If $V(t) \approx 0$: This is the ideal actuarial equilibrium, where the premium has been calibrated precisely to match the expected future obligations. It satisfies the principle of equivalence, which states that the present value of benefits and costs should equal the present value of premiums in a fairly priced insurance scheme.
\end{itemize}

In our simulations, this actuarial balance is achieved when the premium is set to its optimal value $\pi = \pi^*$. Deviating from this optimal rate, either by underpricing ($\pi^* - \varepsilon$) or overpricing ($\pi^* + \varepsilon$), results in non-zero terminal reserves, corresponding to either a funding gap or an actuarial surplus. These deviations highlight the importance of precise premium calibration to ensure fairness, solvency, and consistency with actuarial principles.

\section{Conclusion}\label{S6}
In this paper, we developed an actuarial framework that incorporates epidemic dynamics into life insurance modeling using the SEIARD compartmental model. 
By explicitly accounting for incubation periods, asymptomatic transmission, and disease-induced mortality, our approach extends traditional actuarial models to better reflect the temporal nature of infectious disease outbreaks. 
We derived closed-form expressions for key actuarial quantities such as premiums, annuity benefits, and reserve functions, all dynamically adjusted to reflect the evolving epidemiological environment.

Numerical simulations, carried out through an NSFD scheme, demonstrated the interplay between epidemiological parameters and financial outcomes. 
The SEIARD model~\eqref{E3.1} revealed a high prevalence of asymptomatic transmission, contributing significantly to the infection pressure. 
The forces of infection, mortality, and removal exhibited dynamic trajectories with notable phase lags, highlighting the timing mismatch between exposure and outcome. 
Furthermore, our analysis of infection-free and epidemic-adjusted survival probabilities showed that while a large fraction of the population may contract the disease, mortality remains comparatively low under the chosen parameters.

From an actuarial perspective, the reserve function $V(t)$ provided a time-resolved view of the insurer's liability throughout the epidemic. 
Simulations showed that the reserve is highly sensitive to small deviations in the premium rate. 
In particular, the optimal premium $\pi^*$ achieved actuarial balance by ensuring that the reserve converges to zero at the end of the coverage period. 
Deviations from this rate resulted in either surplus accumulation or potential solvency deficits, emphasizing the importance of dynamic premium calibration under epidemic risk.

\section*{Declarations}

\subsection*{Data availability} 
All information analyzed or generated, which would support the results of this work are available in this article.
No data was used for the research described in the article.

\subsection*{Conflict of interest} 
The authors declare that there are no problems or conflicts 
of interest between them that may affect the study in this paper.

\nocite{Nkeki2024}
\nocite{Chatterjee2008}
\nocite{Macdonald2005}
\nocite{amrullah2025actuarial}
\nocite{Feng2021}
\nocite{Chernov2021}
\bibliographystyle{unsrt}
\bibliography{References}


\end{document}